%% file: main.tex
\documentclass[a4paper]{article}

\usepackage[english]{babel}
\usepackage{amsmath,amsthm,amssymb,amsfonts}
\usepackage{color}
\usepackage{graphicx}
\usepackage{subcaption}
\usepackage{url}
\usepackage{hyperref}
\usepackage{authblk}

\input{macros.tex}

\newcommand{\set}[1]{\left\{ #1 \right\}}

\newtheorem{theorem}{Theorem}

\newtheorem{lemma}{Lemma}

\usepackage[norelsize,ruled,vlined]{algorithm2e}

\begin{document}

\title{The Complexity of Learning \\Linear Temporal Formulas from Examples}
\author[1,2]{Nathana{\"e}l Fijalkow}
\author[1]{Guillaume Lagarde}
\affil[1]{CNRS, LaBRI, Universit{\'e} de Bordeaux, France}
\affil[2]{The Alan Turing Institute in London, United Kingdom}
\date{}

\maketitle

\begin{abstract}
\input{abstract}

\end{abstract}

\section{Introduction}
\input{intro}

\section{Preliminaries}
\input{definitions}

\newpage
\section{$\LTL(\X,\wedge)$}
\label{sec:X_and}
\input{x_and}

\newpage
\section{$\LTL(\F,\wedge)$}
\label{sec:F_and}
\input{f_and}

\newpage
\section{$\LTL(\F,\X,\wedge,\vee)$}
\label{sec:F_X_and_or}
\input{f_x_and_or}

\newpage
\section{Dual results and open problems}
\label{sec:conclusions}
\input{conclusions}

\section*{Acknowledgments}
We thank Daniel Neider for introducing us to this fascinating problem.

\bibliographystyle{alpha}
\bibliography{bib}

\end{document}

%% file: macros.tex
\usepackage{complexity}
\usepackage{framed}

\newcommand{\N}{\mathbb{N}}

\newcommand{\LTL}{\textbf{LTL}}

\newcommand{\X}{\textbf{X}}
\renewcommand{\G}{\textbf{G}}
\newcommand{\F}{\textbf{F}}
\newcommand{\U}{\textbf{U}}
\renewcommand{\exp}{\text{exp}}
\newcommand{\sem}[1]{\langle #1 \rangle}
\newcommand{\Card}{\text{Card}}

\newcommand{\last}{\textbf{last}}
\newcommand{\width}{\textbf{width}}

\newcommand{\RR}{\textbf{R}}
\newcommand{\ind}{\textbf{ind}}

\newtheorem{open}{Open problem}

%% file: abstract.tex
In this paper we initiate the study of the computational complexity of learning linear temporal logic (LTL) formulas from examples.
We construct approximation algorithms for fragments of LTL and prove hardness results; 
in particular we obtain tight bounds for approximation of the fragment containing only the next operator and conjunctions,
and prove $\NP$-completeness results for many fragments.

%% file: intro.tex
We are in this paper interested in the complexity of learning formulas of Linear Temporal Logic ($\LTL$) from examples,
in a passive scenario: from a set of positive and negative words, the objective is to construct a formula, as small as possible,
which satisfies the positive words and does not satisfy the negative words.

Passive learning of languages has a long history paved with negative results.
Learning automata is notoriously difficult from a theoretical perspective, 
as witnessed by the original $\NP$-hardness result of learning a Deterministic Finite Automaton (DFA) from examples~(\cite{Gold78}).
This line of hardness results culminates with the inapproximability result of~\cite{PittW93} stating that 
there is no polynomial time algorithm for learning a DFA from examples even up to a polynomial approximation of their size.

\vskip1em
One approach to cope with such hardness results is to change representation, for instance replacing automata by logical formulas; 
their syntactic structures make them more amenable to principled search algorithms.
There is a range of potential logical formalisms to choose from depending on the application domain.
Linear Temporal Logic~(\cite{Pnueli77}) is a prominent logic for specifying temporal properties over words,
it has become a de facto standard in many fields such as model checking, program analysis, and motion planning for robotics.
A key property making $\LTL$ a strong candidate as a concept class is that its syntax does not include variables, 
contributing to the fact that $\LTL$ formulas are typically easy to interpret and therefore useful as explanations.

\vskip1em
Over the past five to ten years learning temporal logics (of which $\LTL$ is the core) has become an active research area,
with applications in program specification~(\cite{LPB15}) and anomaly and fault detections~(\cite{BoVaPeBe-HSCC-2016}).
A number of different approaches have been proposed, leveraging SAT solvers~(\cite{NeiderG18}), automata~(\cite{Camacho_McIlraith_2019}),
and Bayesian inference~(\cite{ijcai2019-0776}), and extended to more expressive logics such as Property Specification Language (PSL)~(\cite{RoyFismanNeider20}) and Computational Tree Logic (CTL)~(\cite{EhlersGN20}).

\vskip1em
Very little is known about the computational complexity of the underlying problem;
indeed the works cited above focussed on constructing efficient algorithms for practical applications.
The goal of this paper is to initiate the study of the complexity of learning $\LTL$ formulas from examples.

\vskip1em
\textbf{Our contributions.}
We present a set of results for three fragments of $\LTL$.
For all three fragments we show that the learning problem is $\NP$-complete.
\begin{itemize}
	\item In Section~\ref{sec:X_and} we study $\LTL(\X,\wedge)$, which is the fragment containing only the next operator and conjunctions.
	We obtain matching upper and lower bounds on approximation algorithms: 
	we show that there exists a polynomial time $\log(n)$-approximation algorithm for learning $\LTL(\X,\wedge)$,
	and that the approximation ratio cannot be improved for polynomial time algorithms.
	\item In Section~\ref{sec:F_and} we study $\LTL(\F,\wedge)$, which is the fragment containing only the eventually operator and conjunctions.
	We construct an $n$-approximation algorithm 
	and show that there is no polynomial time $\log(n)$-approximation algorithm.
	\item In Section~\ref{sec:F_X_and_or} we study $\LTL(\F,\X,\wedge,\vee)$, which is the fragment containing the eventually and next operators,
	conjunctions and disjunctions.
\end{itemize}
We conclude in Section~\ref{sec:conclusions}, listing remaining open problems.

%% file: definitions.tex
Unless otherwise specified we use the alphabet $\Sigma = \set{a,b}$ of size $2$.
We index words from position $1$ (not $0$) and the letter at position $i$ in the word $w$ is $w(i)$,
so $w = w(1) \dots w(\ell)$.
The empty word is $\varepsilon$.

The syntax\footnote{$\LTL$ also includes an Until operator $\U$ extending both $\F$ and $\G$. In this paper we only consider fragments of $\LTL(\G, \F,\X,\wedge,\vee)$.} of Linear Temporal Logic ($\LTL$) includes atomic formulas $c \in \Sigma$, the boolean operators $\wedge$ and $\vee$,
and the temporal operators $\X, \F$, and $\G$.
The semantic of $\LTL$ over finite words is defined inductively over formulas,
through the notation $w,i \models \phi$ where $w \in \Sigma^*$ is a word of length $\ell$, $i \in [1,\ell]$ is a position in $w$, 
and $\phi$ an $\LTL$ formula.
The definition is given below for the atomic formulas and temporal operators $\X, \F$, and $\G$, with boolean operators interpreted as usual.
\begin{itemize}
    \item $w, i \models c$ if $w(i) = c$.
    \item $w, i \models \X \phi$ if $i < \ell$ and $w, i + 1 \models \phi$. It is called the ne$\X$t operator.
    \item $w, i \models \F \phi$ if $w, i' \models \phi$ for some $i' \in [i,\ell]$. It is called the e$\F$entually operator.
    \item $w, i \models \G \phi$ if $w, i' \models \phi$ for all $i' \in [i,\ell]$. It is called the $\G$lobally operator.
\end{itemize}
We then write $w \models \phi$ if $w,1 \models \phi$ and say that $w$ satisfies $\phi$.
We consider fragments of $\LTL$ by specifying which boolean connectives and temporal operators are allowed.
For instance $\LTL(\X,\wedge)$ is the set of all $\LTL$ formulas using only atomic formulas, conjunctions, and the next operator.
The full logic we consider here is $\LTL = \LTL(\F,\G,\X,\land,\lor)$.
The size of a formula is the size of its syntactic tree. 
We say that two formulas are equivalent if they have the same semantics.

\paragraph*{The $\LTL$ learning problem.}
The $\LTL$ learning decision problem is: 
\begin{framed}
\begin{tabular}{ll}
\textbf{INPUT}: & $u_1,\dots,u_n,v_1,\dots,v_m \in \Sigma^*$ and $k \in \N$,\\
\textbf{QUESTION}: & does there exist an $\LTL$ formula $\phi$ of size at most $k$ \\
& such that for all $j \in [1,n]$, we have $u_j \models \phi$, \\
& and for all $j \in [1,m]$, we have $v_j \not\models \phi$?
\end{tabular}
\end{framed}
In that case we say that $\phi$ separates $u_1,\dots,u_n$ from $v_1,\dots,v_m$,
or simply that $\phi$ is a separating formula if the words are clear from the context.
We call $u_1,\dots,u_n$ the positive words, and $v_1,\dots,v_m$ the negative words.
The $\LTL$ learning problem is analogously defined for any fragment of $\LTL$.

\paragraph*{Parameters for complexity analysis.}
Without loss of generality we can assume that $n = m$ (adding duplicate identical words to have an equal number of positive and negative words).
Therefore the three important parameters for the complexity of the $\LTL$ learning problem are:
$n$ the number of words, $\ell$ the maximum length of the words, and $k$ the desired size for the formula.

\paragraph*{Representation.}
The words given as input are represented in a natural way. 
We emphasise a subtelty on the representation of $k$: it can be given in binary (a standard assumption) or in unary.

In the first case, the input size is $O(n \cdot \ell + \log(k))$, 
so the formula $\phi$ we are looking for may be exponential in the input size!
Therefore it is not clear a priori that the $\LTL$ learning problem is in $\NP$.
Opting for a unary encoding, the input size becomes $O(n \cdot \ell + k)$, 
and in that case an easy argument shows that the $\LTL$ learning problem is in $\NP$.

%

We follow the standard representation: $k$ is given in binary, and therefore it is not immediate that the $\LTL$ learning problem is in $\NP$.

\paragraph*{Convention.}
Typically $i \in [1,\ell]$ is a position in a word and $j \in [1,n]$ is used for indexing words.

\paragraph*{A naive algorithm.}
Let us start our complexity analysis of the learning $\LTL$ problem by constructing a naive algorithm for the whole logic.
\begin{theorem}\label{thm:recursive_algorithm}
There exists an algorithm for learning $\LTL$ in time and space $O(\exp(k) \cdot n \cdot \ell)$,
where $\exp(k)$ is exponential in $k$.
\end{theorem}

Notice that the dependence of the algorithm presented in Theorem~\ref{thm:recursive_algorithm} is linear in $n$ and $\ell$, 
and it is exponential only in $k$, but since $k$ is represented in binary this is potentially a doubly-exponential algorithm.

\begin{proof}
For a formula $\phi \in \LTL$, we write $\sem{\phi} : \set{u_1,\dots,u_n,v_1,\dots,v_n} \to \set{0,1}^\ell$ for the function defined by
\[
\sem{\phi}(w)(i) = 
\begin{cases}
1 & \text{ if } w, i \models \phi, \\
0 & \text{ if } w, i \not\models \phi,
\end{cases}
\]
for $w \in \set{u_1,\dots,u_n,v_1,\dots,v_n}$.

Note that $\phi$ is separating if and only if $\sem{\phi}(u_j)(1) = 1$ and $\sem{\phi}(v_j)(1) = 0$ for all $j \in [1,n]$.
The algorithm simply consists in enumerating all formulas $\phi$ of $\LTL$ of size at most $k$ inductively,
constructing $\sem{\phi}$, and checking whether $\phi$ is separating.
Initially, we construct $\sem{a}$ and $\sem{b}$, and then once we have computed 
$\sem{\phi}$ and $\sem{\psi}$, we can compute $\sem{\phi \wedge \psi}$, $\sem{\phi \vee \psi}$,  $\sem{\X \phi}, \sem{\F \phi}$ and $\sem{\G \phi}$ in time $O(n \cdot \ell)$.
To conclude, we note that the number of formulas\footnote{The asymptotics can be obtained using classical techniques from Analytic Combinatorics~\cite{Flajolet08analyticcombinatorics}: the number of $\LTL$ formulas of size $k$ is asymptotically equivalent to 
$\frac{\sqrt{14} \cdot 7^k}{2 \sqrt{\pi k^3}}$.} of $\LTL$ of size at most $k$ is exponential in $k$.
\end{proof}

\paragraph*{Approximation algorithms.}
The goal of this paper is to understand the complexity of learning fragments of $\LTL$ and to construct efficient approximation algorithms.
An $\alpha$-approximation algorithm for learning $\LTL$ (or some fragment of $\LTL$) does the following:
the algorithm either determines that there are no separating formulas, 
or constructs a separating formula $\phi$ which has size at most $\alpha \cdot m$ with $m$ the size of a minimal separating formula.

%% file: x_and.tex
\subsection*{Normalisation}
We first state and prove a normalisation lemma for formulas in $\LTL(\X,\wedge)$.

We define the class of ``patterns'' as formulas generated by the following grammar:
\[
P \doteq \X^i c \ \mid\ \X^i (c \wedge P) \text{ with $i \ge 0$ and $c \in \Sigma$}.
\]
Unravelling the definition we get the following general form for patterns:
\[
P = \X^{i_1 - 1}(c_1 \wedge \X^{i_2 - i_1}(\cdots \wedge \X^{i_p - i_{p-1}} c_p)\cdots),
\]
with $1 \le i_1 < i_2 < \dots < i_p$ and $c_1,\dots,c_p \in \Sigma$.
It is equivalent to the (larger in size) formula $\bigwedge_{q \in [1,p]} \X^{i_q - 1} c_q$, which states that
for each $q \in [1,p]$, the letter in position $i_q$ is $c_q$.

To determine the size of a pattern $P$ we look at two parameters: 
its last position $\last(P) = i_p$ and its width $\width(P) = p$.
The size of $P$ is $\last(P) + 2 (\width(P) - 1)$.
The two parameters of a pattern, last position and width, hint at the key trade-off we will have to face in learning $\LTL(\X,\wedge)$ formulas:
do we increase the last position, to reach further letters in the words, or the width, to further restrict the set of satisfying words?

\begin{lemma}\label{lem:normalisation_X_and}
For every formula $\phi \in \LTL(\X,\wedge)$ there exists an equivalent pattern of size smaller than or equal to $\phi$.
\end{lemma}
\begin{proof}
We proceed by induction on $\phi$.
\begin{itemize}
	\item Atomic formulas are already a special case of patterns.
	\item If $\phi = \X \phi'$, by induction hypothesis we get a pattern $P$ equivalent to $\phi'$,
	then $\X P$ is a pattern and equivalent to $\phi$. 
	\item If $\phi = \phi_1 \wedge \phi_2$, by induction hypothesis we get two patterns $P_1$ and $P_2$ equivalent to $\phi_1$ and $\phi_2$.
	We use the inductive definition for patterns to show that $P_1 \wedge P_2$ is equivalent to another pattern.
	We focus on the case $P_1 = \X^{i_1} (c_1 \wedge P'_1)$ and $P_2 = \X^{i_2} (c_2 \wedge P'_2)$, 
	the other cases are simpler instances of this one.

	There are two cases: $i_1 = i_2$ or $i_1 \neq i_2$.

	If $i_1 = i_2$, either $c_1 \neq c_2$ and then $P_1 \wedge P_2$ is equivalent to false, which is the pattern $c_1 \wedge c_2$,
	or $c_1 = c_2$, and then $P_1 \wedge P_2$ is equivalent to $\X^{i_1} (c_1 \wedge P'_1 \wedge P'_2)$.
	By induction hypothesis $P'_1 \wedge P'_2$ is equivalent to a pattern $P'$, so the pattern $\X^{i_1} (c_1 \wedge P')$
	is equivalent to $P_1 \wedge P_2$, hence to $\phi$.

	If $i_1 \neq i_2$, without loss of generality $i_1 < i_2$, 
	then $P_1 \wedge P_2$ is equivalent to 
	$\X^{i_1} (c_1 \wedge P'_1 \wedge \X^{i_2 - i_1} (c_2 \wedge P'_2))$.
	By induction hypothesis $P'_1 \wedge \X^{i_2 - i_1} (c_2 \wedge P'_2)$ is equivalent to a pattern $P'$, 
	so the pattern $\X^{i_1} (c_1 \wedge P')$.
	is equivalent to $P_1 \wedge P_2$, hence to $\phi$.
\end{itemize}
\end{proof}

The first simple corollary of Lemma~\ref{lem:normalisation_X_and} is a non-deterministic polynomial time algorithm.

\begin{theorem}\label{thm:ltl_X_and_NP}
The learning problem for $\LTL(\X,\land)$ is in $\NP$.
\end{theorem}

\begin{proof}
Let $u_1,\dots,u_n,v_1,\dots,v_n$ a set of $2n$ words of length at most $\ell$.
Thanks to Lemma~\ref{lem:normalisation_X_and}, if there exists a separating formula $\phi$, then there exists 
a separating pattern of size no larger than $\phi$.
However patterns have polynomially bounded size: indeed both the last position and the width are at most $\ell$,
so the size of a pattern is at most $3\ell - 2 = O(\ell)$.

In other words, if there exists a separating formula, then there exists one of size linear in $\ell$.
A non-deterministic algorithm guesses such a formula and checks whether it is indeed separating in (deterministic) time $O(n \cdot \ell^2)$.
\end{proof}

\subsection*{An approximation algorithm}
\begin{theorem}\label{thm:approximation_algorithm_ltl_X_and}
There exists a $O(n \cdot \ell^2)$ time $\log(n)$-approximation algorithm for learning $\LTL(\X,\land)$.
\end{theorem}

\begin{algorithm}[ht]
\KwData{Words $u_1,\dots,u_n, v_1,\dots,v_n$ of length at most $\ell$.}

$X \leftarrow \set{i \in [1,\ell] : \exists c \in \Sigma, \forall j \in [1,n], u_j(i) = c}$

\For{$i \in X$}{
	$Y_i \leftarrow \set{j \in [1,n] : v_j(i) \neq u_1(i) = u_2(i) = \dots = u_n(i)}$
}

$I_0 \leftarrow \emptyset$

$C_0 \leftarrow \emptyset$

$x \leftarrow 0$

\Repeat{$C_x = [1,n] \text{ or } I_x = X$}{
	$i_x \leftarrow \text{argmax} \set{ \Card(Y_i \setminus C_x) : i \in X \setminus I_x}$ ;

	$I_{x+1} \leftarrow I_x \cup \set{i_x}$ ;

	$C_{x+1} \leftarrow C_x \cup Y_{i_x}$ ;

	$x \leftarrow x + 1$ ;
}

\If{$C_x = [1,n]$}{
	\Return{\text{The pattern corresponding to $I_x$}}
}
\Else{
	\Return{No separating formula}
}

\caption{\label{algo:greedy_algorithm} The greedy algorithm returning a $\log(n)$-approximation of a minimal separating $\LTL(\X,\wedge)$-formula
with last position $\ell$.}
\end{algorithm}

\begin{proof}
Let $u_1,\dots,u_n,v_1,\dots,v_n$ a set of $2n$ words of length at most $\ell$.
Thanks to Lemma~\ref{lem:normalisation_X_and} we are looking for a separating pattern:
\[
P = \X^{i_1 - 1}(c_1 \land \X^{i_2 - i_1}(\cdots \land \X^{i_p - i_{p-1}} c_p)\cdots).
\]
For a pattern $P$ we define $I(P) = \set{i_q \in [1,\ell] : q \in [1,p]}$.
Note that $\last(P) = \max I(P)$ and $\width(P) = \Card(I(P))$.

We define the set 
$X = \set{i \in [1,\ell] : \exists c \in \Sigma, \forall j \in [1,n], u_j(i) = c}$.
Note that $P$ satisfies $u_1,\dots,u_n$ if and only if $I(P) \subseteq X$.
Further, given $I \subseteq X$, we can construct a pattern $P$ such that $I(P) = I$ and $P$ satisfies $u_1,\dots,u_n$: 
we simply choose $c_q = u_1(i_q) = \dots = u_n(i_q)$ for $q \in [1,p]$.
We call $P$ the pattern corresponding to $I$.

Recall that the size of the pattern $P$ is $\last(P) + 2(\width(P) - 1)$.
This makes the task of minimising it difficult: 
there is a trade-off between minimising the last position $\last(P)$ and the width $\width(P)$.

Let us consider the following easier problem: 
construct a $\log(n)$-approximation of a minimal separating pattern with fixed last position.
Assuming we have such an algorithm, we obtain a $\log(n)$-approximation of a minimal separating pattern
by running the previous algorithm on prefixes of length $\ell'$ for each $\ell' \in [1,\ell]$.

\vskip1em
We now focus on the question of constructing a $\log(n)$-approximation of a minimal separating pattern with fixed last position.
We refer to Algorithm~\ref{algo:greedy_algorithm} for the pseudocode.
For a set $I$, we write $C_I = \bigcup \set{Y_i : i \in I}$: 
the pattern corresponding to $I$ does not satisfy $v_j$ if and only if $j \in C_I$.
In particular, the pattern corresponding to $I$ is separating if and only if $C_I = [1,n]$.
 
The algorithm constructs a set $I$ incrementally through the sequence $(I_x)_{x \ge 0}$, 
with the following easy invariant: for $x \ge 0$, we have $C_x = C_{I_x}$.
The algorithm is greedy: $I_x$ is augmented with $i \in X \setminus I_x$ 
maximising the number of words added to $C_x$ by adding $i$, which is the cardinality of $Y_i \setminus C_x$.

\vskip1em
We now prove that this yields a $\log(n)$-approximation algorithm.
Let $P_{\text{opt}}$ a minimal separating pattern with last position $\ell$,
inducing $I_{\text{opt}} = I(P_{\text{opt}}) \subseteq [1,\ell]$ of cardinal $m$.
Note that $C_{I_{\text{opt}}} = [1,n]$.

We let $n_x = n - |C_x|$ and show the following by induction on $x \ge 0$: 
\[
n_{x+1} \le n_x \cdot \left( 1 - \frac{1}{m} \right) = n_x \cdot \frac{m - 1}{m}.
\]

We claim that there exists $i \in X \setminus I_x$ such that $\Card(Y_i \setminus C_x) \ge \frac{n_x}{m}$.
Indeed, assume towards contradiction that for all $i \in X \setminus I_x$ we have $\Card(Y_i \setminus C_x) < \frac{n_x}{m}$,
then there are no sets $I$ of cardinal $m$ such that $C_I \supseteq [1,n] \setminus C_x$,
contradicting the existence of $I_{\text{opt}}$.
Thus there exists $i \in X \setminus I_x$ such that $\Card(Y_i \setminus C_x) \ge \frac{n_x}{m}$,
implying that the algorithm chooses such an $i$ and 
$n_{x + 1} \le n_x - \frac{n_x}{m} = n_x \cdot \left( 1 - \frac{1}{m} \right)$.

\vskip1em
The proved inequality implies $n_x \le n \cdot \left( 1 - \frac{1}{m} \right)^x$.
This quantity is less than $1$ for $x \ge \log(n) \cdot m$, implying that the algorithm
stops after at most $\log(n) \cdot m$ steps.
Consequently, the pattern corresponding to $I$ has size at most $\log(n) \cdot |P_{\text{opt}}|$,
completing the claim on approximation.

\vskip1em
A naive complexity analysis yields an implementation of Algorithm~\ref{algo:greedy_algorithm} running in time $O(n \cdot \ell)$,
leading to an overall complexity of $O(n \cdot \ell^2)$ by running Algorithm~\ref{algo:greedy_algorithm} 
on the prefixes of length $\ell'$ of $u_1,\dots,u_n,v_1,\dots,v_n$ for each $\ell' \in [1,\ell]$.
\end{proof}

\subsection*{Hardness results}
\begin{theorem}\label{thm:hardness_X_and}
The $\LTL(\X, \wedge)$ learning problem is $\NP$-hard, 
and there are no $(1 - o(1)) \cdot \log(n)$ polynomial time approximation algorithms unless $\P = \NP$,
even for a single positive word.
\end{theorem}

Note that Theorem~\ref{thm:approximation_algorithm_ltl_X_and} and Theorem~\ref{thm:hardness_X_and} 
yield matching upper and lower bounds on approximation algorithms for learning $\LTL(\X,\land)$.

The hardness result stated in Theorem~\ref{thm:hardness_X_and} follows from a reduction to the set cover problem,
that we define now.
The set cover decision problem is: given $S_1,\dots,S_\ell$ subsets of $[1,n]$ and $k \in \N$,
does there exists $I \subseteq [1,\ell]$ of size at most $k$ such that $\bigcup_{i \in I} S_i = [1,n]$?
In that case we say that $I$ is a cover. 
An $\alpha$-approximation algorithm returns a cover of size at most $\alpha \cdot k$ where $k$ is the size of a minimal cover.
The following results form the state of the art for solving exact and approximate variants of the set cover problem.

\begin{theorem}[\cite{DinurS14}]
\label{thm:subset_cover}
The set cover problem is $\NP$-complete, 
and there are no $(1 - o(1)) \cdot \log(n)$ polynomial time approximation algorithms unless $\P = \NP$.
\end{theorem}

\begin{proof}
We construct a reduction from set cover. 
Let $S_1,\dots,S_\ell$ subsets of $[1,n]$ and $k \in \N$.

Let us consider the word $u = a^{\ell + 1}$,
and for each $j \in [1,n]$ and $i \in [1,\ell]$, writing $v_j(i)$ for the $i$\textsuperscript{th} letter of $v_j$:
\[
v_j(i) = 
\begin{cases}
b \text{ if } j \in S_i, \\
a \text{ if } j \notin S_i,\\
\end{cases}
\]
and we set $v_{j}(\ell+1)=a$ for any $j \in [1,n]$. We also add $v_{n+1} = a^\ell b$.

We claim that there is a cover of size $k$ if and only if 
there is a formula of size $\ell + 2k - 1$ separating $u$ from $v_1,\dots,v_{n+1}$.

Thanks to Lemma~\ref{lem:normalisation_X_and} we can restrict our attention to patterns, \textit{i.e} formulas of the form (we adjust the indexing for technical convenience)
\[
\phi = \X^{i_1 - 1}(c_1 \wedge \X^{i_2 - i_1}(\cdots \wedge \X^{i_{p+1} - i_p} c_{p+1})\cdots),
\]
for some positions $i_1 \le \dots \le i_{p+1}$ and letters $c_1,\dots,c_{p+1} \in \Sigma$.
If $\phi$ satisfies $u$, then necessarily $c_1 = \dots = c_{p+1} = a$.
This implies that if $\phi$ does not satisfy $v_{n+1}$, then necessarily $i_{p+1} = \ell + 1$.

We associate to $\phi$ the set $I = \set{i_1 \le \dots \le i_p}$.
Note that $\phi$ is equivalent to $\bigwedge_{q \in [1,p]} \X^{i_q - 1} a \wedge \X^{\ell} a$, and the size of $\phi$ is $\ell + 1 + 2 (|I|-1)$.

By construction, $\phi$ separates $u$ from $v_1,\dots,v_{n+1}$ if and only if $I$ is a cover.
Indeed, $I$ is a cover if and only if for every $j \in [1,n]$ there exists $i \in I$ such that $j \in S_i$,
which is equivalent to 
for every $j \in [1,n]$ we have $v_j \not\models \phi$.
\end{proof}

%% file: f_and.tex
As we will see, $\LTL(\F,\wedge)$ over an alphabet of size $2$ is very weak.
This degeneracy vanishes when considering alphabets of size at least $3$.
Let us fix a (finite) alphabet $\Sigma$.

\subsection*{Minimal formulas}
Instead of defining a normal form as we did for $\LTL(\X,\wedge)$ 
we characterise the expressive power of $\LTL(\F,\wedge)$ and construct for each property expressible in this logic a minimal formula.

Let us consider two words $u = u(1) \dots u(\ell')$ and $v = v(1) \dots v(\ell)$.
We say that $u$ is a subword of $v$ if there exists $\phi : [1,\ell'] \to [1,\ell]$ increasing such that $v(\phi(i)) = u(i)$,
and that $u$ is a factor of $v$ if $v = v_1 u v_2$ for two words $v_1,v_2$.
For example $abba$ is a subword of $b\textbf{ab}aaaa\textbf{b}\textbf{a}$, but not a factor,
and $bba$ is a factor of $ba\textbf{bba}a$.
We say that a word is non-repeating if every two consecutive letters are different.

\begin{lemma}\label{lem:characterisation_f_and}
For every formula $\phi \in \LTL(\F,\wedge)$,
either it is equivalent to false or there exists a finite set of non-repeating words $w_1,\dots,w_p$ and $c \in \Sigma \cup \set{\varepsilon}$
such that for every word $z$,
\[
z \models \phi \text{ if and only if }
\begin{cases}
\text{for all $q \in [1,p], w_q$ is a subword of $z$}, \\
\text{and $z$ starts with $c$}.
\end{cases}
\]
\end{lemma}

\begin{proof}
We proceed by induction over $\phi$.
\begin{itemize}
	\item For the atomic formula $c \in \Sigma$, the property is satisfied using the empty set of words and $c$.
	\item If $\phi = \F \phi'$, by induction hypothesis we get $w_1,\dots,w_p$ and $c$ for $\phi'$.
        We let $w'_i = cw_i$ if $w_i(1) \neq c$ and $w'_i = w_i$ otherwise, then $z \models \phi$ if and only if 
        for all $q \in [1,p]$, $w'_q$ is a subword of $z$ and $z$ starts with $\varepsilon$ (the latter condition is always satisfied).
	\item If $\phi = \phi_1 \wedge \phi_2$, by induction hypothesis we get $w^1_1,\dots,w^1_{p_1}, c_1$ for $\phi_1$
	and $w^2_1,\dots,w^2_{p_2}, c_2$ for $\phi_2$.
	There are two cases.
	If $c_1$ and $c_2$ are non-empty and $c_1 \neq c_2$ then $\phi$ is equivalent to false.
	Otherwise, either both are non-empty and equal or at least one is $\varepsilon$, say $c_2$.
	In both cases,
	$u \models \phi$ if and only if for all $(e,q) \in (1,[1,p_1]) \cup (2,[1,p_2])$, $w^e_q$ is a subword of $u$ and $u$ starts with $c_1$.
\end{itemize}
\end{proof}

Lemma~\ref{lem:characterisation_f_and} gives a characterisation of the properties expressible in $\LTL(\F,\wedge)$.
It implies that over an alphabet of size $2$ the fragment $\LTL(\F,\wedge)$ is very weak.
Indeed, there are very few non-repeating words over the alphabet $\Sigma = \set{a,b}$:
only prefixes of $abab \dots$ and $baba \dots$.
This implies that formulas in $\LTL(\F,\wedge)$ over $\Sigma = \set{a,b}$
can only place lower bounds on the number of alternations between $a$ and $b$ (starting from $a$ or from $b$)
and check whether the word starts with $a$ or $b$.
In particular, the $\LTL(\F,\wedge)$ learning problem over this alphabet is (almost) trivial and thus not interesting.
Hence we now assume that $\Sigma$ has size at least $3$.

\vskip1em
We move back from semantics to syntax, and show how to construct minimal formulas.
Let $w_1,\dots,w_p$ a finite set of non-repeating words and $c \in \Sigma \cup \set{\varepsilon}$, 
we define a formula $\phi$ as follows.

The set of prefixes of $w_1,\dots,w_p$ are organised in a forest (set of trees): 
a node is labelled by a prefix $w$ of some $w_1,\dots,w_p$,
and its children are the words $wc$ which are prefixes of some $w_1,\dots,w_p$.
The leaves are labelled by $w_1,\dots,w_p$.
We interpret each tree $t$ as a formula $\phi_t$ in $\LTL(\F,\wedge)$ as follows, in an inductive fashion: for $c \in \Sigma$,
if $t$ is labelled $w a$ with subtrees $t_1,\dots,t_q$, then
\[
\phi_t = \F( c \wedge \bigwedge_i \phi_{t_i}).
\]
If $c = \varepsilon$, the formula associated to $w_1,\dots,w_p$ and $c$ is the conjunction of the formulas for each tree of the forest,
and if $c \in \Sigma$, then the formula additionally has a conjunct $c$.

As an example, consider the set of words $ab, ac, bab$, and the letter $a$.
The forest corresponding to $ab, ac, bab$ contains two trees: 
one contains the nodes $b, ba, bab$, and the other one the nodes $a, ab, ac$.
The two corresponding formulas are
\[
\F (b \wedge \F(a \wedge \F b)) \qquad ; \qquad \F(a \wedge \F b \wedge \F c).
\]
And the formula corresponding to the set of words $ab, ac, bab$, and the letter $a$ is
\[
a\ \wedge \ \F (b \wedge \F(a \wedge \F b)) \ \wedge\ \F(a \wedge \F b \wedge \F c).
\]

\begin{lemma}\label{lem:minimal_f_and}
For every non-repeating words $w_1,\dots,w_p$ and $c \in \Sigma \cup \set{\varepsilon}$,
the formula $\phi$ constructed above is minimal, meaning there are no smaller equivalent formulas.
\end{lemma}

Applying the construction above to a single non-repeating word $w = c_1 \dots c_p$ we obtain what we call a ``\textbf{f}attern'' 
(pattern with an F):
\[
F = \F (c_1 \wedge \F(\cdots \wedge \F c_p)\cdots),
\]
We say that the non-repeating word $w$ induces the fattern $F$ above, 
and conversely that the fattern $F$ induces the word $w$.
The size of a fattern $F$ is $3 |w| - 1$. 
Adding the initial letter we obtain a grounded fattern $c \wedge F$,
in that case the letter $c$ is added at the beginning of $w$ and the size is $3 |w| - 2$.

\begin{lemma}\label{lem:normalisation_F_and}
Let $u_1,\dots,u_n,v_1,\dots,v_n$.
If there exists $\phi \in \LTL(\F,\wedge)$ separating $u_1,\dots,u_n$ from $v_1,\dots,v_n$,
then there exists a conjunction of at most $n$ fatterns separating $u_1,\dots,u_n$ from $v_1,\dots,v_n$.
\end{lemma}
\begin{proof}
Thanks to Lemma~\ref{lem:characterisation_f_and}, to the separating formula $\phi$ 
we can associate a finite set of non-repeating words $w_1,\dots,w_p$ and $c \in \Sigma \cup \set{\varepsilon}$
such that for every word $z$,
\[
z \models \phi \text{ if and only if }
\begin{cases}
\text{for all $q \in [1,p], w_q$ is a subword of $z$}, \\
\text{and $z$ starts with $c$}.
\end{cases}
\]
Let $j \in [1,n]$, since $v_j$ does not satisfy $\phi$ either $v_j$ does not start with $c$
or for some $q \in [1,p]$ the word $w_q$ is not a subword of $u$.
For each $j \in [1,n]$ such that $v_j$ starts with $c$, we pick one $q_j \in [1,p]$ for which $w_{q_j}$ is not a subword of $v_j$,
and consider the set $\set{w_{q_j} : j \in [1,n]}$ together with $c \in \Sigma \cup \set{\varepsilon}$.
The formula induced by the construction above is a conjunction of at most $n$ fatterns and it separates $u_1,\dots,u_n$
from $v_1,\dots,v_n$.
\end{proof}

\vskip1em
A first corollary of Lemmas~\ref{lem:characterisation_f_and} and~\ref{lem:minimal_f_and} is a non-deterministic polynomial time algorithm.

\begin{theorem}\label{thm:ltl_F_and_NP}
The learning problem for $\LTL(\F,\land)$ is in $\NP$.
\end{theorem}

\begin{proof}
Let $u_1,\dots,u_n,v_1,\dots,v_n$ a set of $2n$ words of length at most $\ell$.
Assume there exists a separating formula $\phi$, thanks to Lemma~\ref{lem:normalisation_F_and}
there exists a conjunction of at most $n$ fatterns separating $u_1,\dots,u_n$ from $v_1,\dots,v_n$.
However fatterns have polynomially bounded size: indeed 
the size of a fattern is at most $3 \ell - 1 = O(\ell)$.

In other words, if there exists a separating formula, then there exists one of size at most $O(n \cdot \ell)$.
A non-deterministic algorithm guesses such a formula and checks whether it is indeed separating in (deterministic) time $O(n^2 \cdot \ell)$.
\end{proof}

\subsection*{A dynamic programming algorithm}
Let us define an intermediate problem called shortest subword:
the input is $u_1,\dots,u_n,v_1,\dots,v_n$, 
and the goal is to find the shortest word $w$ such that
for all $j \in [1,n]$, $w$ is a subword of $u_j$ and not a subword of $v_j$.

Lemma~\ref{lem:normalisation_F_and} and Lemma~\ref{lem:normalisation_F_and_single_positive} imply that
learning $\LTL(\F,\wedge)$ in both cases of a single positive word and a single negative word
is equivalent to the shortest subword problem, since minimising the size of a flattern is equivalent to minimising the size of the word it induces.
In particular, this implies that the shortest subword problem is $\NP$-complete.
Let us construct an algorithm for solving the shortest subword problem 
and then discuss its consequences for learning $\LTL(\F,\wedge)$.

\begin{lemma}\label{lem:shortest_subword}
There exists an algorithm solving the shortest subword problem running in time 
$O(n \cdot (\min \set{2^{\ell}, \ell^{2n}} + |\Sigma| \cdot \ell))$.
\end{lemma}

\begin{algorithm}[ht]
\KwData{Words $u_1,\dots,u_n, v_1,\dots,v_n$ of length at most $\ell$.}

\For{$(i_1,\dots,i_n)$}{
	$\RR(i_1, \dots, i_n, \omega, \dots, \omega) \leftarrow 0$ ;
}

\For{$\overline{i}$ with $i'_j = \omega$}{
	$\RR(\overline{i}) \leftarrow \omega$ ;
}

\For{$w \in \{u_1, \dots, v_n \}, c \in \Sigma, i$}{
	$\ind(w, c, i) \leftarrow \min \set{i' : w(i') = c \wedge i' \ge i}$
}

\For{$\overline{i} = (i_1, \dots i_n, i'_1, \dots, i'_n)$}{
	\For{$j \in [1,n]$}{
		$ni_j = \ind(u_j,u_1(i_1), i_j)$ ;
	
		$ni'_j = \ind(v_j,u_1(i_1), i'_j)$ ;
	}
	
	$x \leftarrow \RR(i_1 + 1, ni_2,  \dots, ni_n, ni'_1,\dots, ni'_n)$ ;
	
	$y \leftarrow \RR(i_1 + 1, i_2, \dots, i_n, i'_1, \dots, i'_n))$ ;

	$\RR(\overline{i}) \leftarrow \min(1 + x, y)$ ;	
}

\Return{$\RR(1,\dots,1)$}

\caption{\label{algo:dynamic_algorithm} The dynamic programming algorithm solving the shortest subword problem.}
\end{algorithm}

We use Python-inspired notations for suffixes: 
we let $w(k:)$ denote the word obtained from $w$ starting at position $k$.

Let us write $\overline{i} = (i_1,i_2,\dots,i_n, i'_1,i'_2,\dots,i'_n)$ for a tuple of positions in each of the $2n$ words.
We include for each word the special position $\omega$.
Let $\RR(\overline{i})$ be the length of a shortest word $w$ such that 
for all $j \in [1,n]$, $w$ is a subword of $u_j(i_j:)$ and not a subword of $v_j(i'_j:)$.
We construct a dynamic programming algorithm populating the table $\RR$;
the goal is to compute $\RR(1,1,\dots,1)$. 
The pseudocode is given in Algorithm~\ref{algo:dynamic_algorithm};
we note that it only computes the length of a shortest word, not a word itself.
The algorithm can be easily adapted to output such a word using classical techniques for dynamic programming algorithms.

\begin{proof}
The key equality on which Algorithm~\ref{algo:dynamic_algorithm} relies is
\[
\RR(\overline{i}) = \min
\begin{cases}
1 + \RR(i_1 + 1, ni_2,  \dots, ni_n, ni'_1,\dots, ni'_n)\\
\RR(i_1 + 1, i_2, \dots, i_n, i'_1, \dots, i'_n)
\end{cases},
\]	
where $ni_j = \ind(u_j,u_1(i_1), i_j)$ and $ni'_j = \ind(v_j,u_1(i_1), i'_j)$.
It corresponds to the following case distinction:
we consider the shortest subword $w$ from $\overline{i}$ together 
with the functions $\phi_i$ mapping $w$ to each $u_1,\dots,u_n,v_1,\dots,v_n$.
Then
\begin{itemize}
	\item either $\phi_1(1) = i_1$, and then necessarily $\phi_i(1) \ge ni_i$ for $i \in [2,n]$ and $\phi_{i'}(1) \ge ni'_i$,
	so $w(2:)$ is the shortest subword starting from 
	\[
	(i_1 + 1, ni_2,  \dots, ni_n, ni'_1,\dots, ni'_n),
	\]
	\item or $\phi_1(1) > i_1$, and then $w$ 
	is the shortest subword starting from 
	\[
	(i_1 + 1, i_2, \dots, i_n, i'_1, \dots, i'_n).
	\]
\end{itemize}

Complexity analysis.  There are at most $2^{\ell}$ subwords, and at most $\ell^{2n}$ tuples;
both give an upper bound on the number of iterations.
Processing each is done in time $O(n)$ since
we need to query the values $\ind(u_j,u_1(i_1), i_j)$ and $\ind(v_j,u_1(i_1), i'_j)$ for $j\in [1,n]$.  
The naive algorithm to compute all the values $\ind(w,c,i)$ runs in time $O(n \cdot |\Sigma| \cdot \ell^2)$ 
but this can be easily reduced to a running time of $O(n\cdot |\Sigma| \cdot \ell)$.
\end{proof}

We now show how to instantiate Algorithm~\ref{algo:dynamic_algorithm} for learning $\LTL(\F,\wedge)$.

\begin{theorem}\label{thm:algorithm_f_and}
\hfill
\begin{itemize}
	\item There exists a $O(n \cdot (\min \set{2^{\ell}, \ell^{n}} + |\Sigma| \cdot \ell))$ time algorithm for learning $\LTL(\F,\wedge)$ with a single negative word.
	\item There exists a $O(n \cdot (\min \set{2^{\ell}, \ell^{n}}  + |\Sigma| \cdot \ell))$ time algorithm for learning $\LTL(\F,\wedge)$ with a single positive word.
	\item There exists a $O(n^2 \cdot (\min \set{2^{\ell}, \ell^{n}} + |\Sigma| \cdot \ell)))$ time $n$-approximation algorithm for learning $\LTL(\F,\wedge)$.
\end{itemize}
\end{theorem}

\begin{proof}
Let us first consider the case of a single negative word. 
Thanks to Lemma~\ref{lem:normalisation_F_and} we can restrict our attention to fatterns,
so in this case learning $\LTL(\F,\wedge)$ is equivalent to the shortest subword problem with a single negative word.
Instantiating Lemma~\ref{lem:shortest_subword} we obtain
a $O(n \cdot (\min \set{2^{\ell}, \ell^{n+1}} + |\Sigma| \cdot \ell))$ time algorithm for learning $\LTL(\F,\wedge)$ with a single negative word.

The case of a single positive word is similar, 
but invoking Lemma~\ref{lem:normalisation_F_and_single_positive} instead of Lemma~\ref{lem:normalisation_F_and}.

\vskip1em
Let us now consider the general problem of learning $\LTL(\F,\wedge)$.
The algorithm is the following: for each $j \in [1,n]$ 
we run the algorithm for learning $\LTL(\F,\wedge)$ on a single negative word:
we construct a formula $\phi_j$ separating $u_1,\dots,u_n$ from $v_j$.
The algorithm then outputs the formula $\psi = \bigwedge_{j \in [1,n]} \phi_j$.

Indeed $\psi$ separates $u_1,\dots,u_n$ from $v_1,\dots,v_n$.
We now claim that $|\psi| \le n \cdot m$ where $m$ is the size of a minimal formula in $\LTL(\F,\wedge)$ 
separating $u_1,\dots,u_n$ from $v_1,\dots,v_n$.
Let $\phi$ such a formula, then for all $j \in [1,n]$ it also separates $u_1, \dots, u_n$ from $v_j$,
so $|\phi_j| \le |\phi|$, implying that $|\psi| \le n \cdot |\phi|$.
\end{proof}

\subsection*{Hardness results}
\begin{theorem}\label{thm:hardness_F_and}
The $\LTL(\F, \wedge)$ learning problem is $\NP$-hard, and there are no 
$(1 - o(1)) \cdot \log(n)$ polynomial time approximation algorithms unless $\P = \NP$,
even with a single positive word.
\end{theorem}

The result follows from a reduction from the hitting set problem.
The hitting set decision problem is: given $C_1,\dots,C_n$ subsets of $[1,\ell]$ and $k \in \N$,
does there exist $H$ subset of $[1,\ell]$ of size at most $k$ such that for every $j \in [1,n]$ we have $H \cap C_j \neq \emptyset$.
In that case we say that $H$ is a hitting set.

The hitting set problem is an equivalent formulation of the set cover problem,
but it is here technically more convenient to construct a reduction from the hitting set problem.
The hardness results stated in Theorem~\ref{thm:subset_cover} apply to the hitting set problem.

\vskip1em
For proving the correction of the reduction we need a normalisation lemma specialised to the case of a single positive word.

\begin{lemma}\label{lem:normalisation_F_and_single_positive}
Let $u,v_1,\dots,v_n$.
If there exists $\phi \in \LTL(\F,\wedge)$ separating $u$ from $v_1,\dots,v_n$,
then there exists a fattern of size smaller than or equal to $\phi$ separating $u$ from $v_1,\dots,v_n$.
\end{lemma}
\begin{proof}
Thanks to Lemma~\ref{lem:characterisation_f_and}, to the separating formula $\phi$ 
we can associate a finite set of non-repeating words $w_1,\dots,w_p$ and $c \in \Sigma \cup \set{\varepsilon}$
such that for every word $z$,
\[
z \models \phi \text{ if and only if }
\begin{cases}
\text{for all $q \in [1,p], w_q$ is a subword of $z$}, \\
\text{and $z$ starts with $c$}.
\end{cases}
\]
Since $u$ satisfies $\phi$, it starts with $c$ and for all $q \in [1,p]$, $w_q$ is a subword of $u$.
For each $q \in [1,p]$ there exists $\phi_q$ mapping the positions of $w_q$ to $u$.
Let us write $w$ for the word obtained by considering all positions mapped by $\phi_q$ for $q \in [1,p]$.
By definition $w$ is a subword of $u$, and for all $q \in [1,p]$ $w_q$ is a subword of $w$.
It follows that the fattern induced by $w$ separates $u$ from $v_1,\dots,v_n$.
The size of $w$ is at most the sum of the sizes of the $w_q$ for $q \in [1,p]$,
hence the fattern induced by $w$ is smaller than the original formula $\phi$.
\end{proof}

We can now prove Theorem~\ref{thm:hardness_F_and}.

\begin{proof}
We construct a reduction from the hitting set problem.
Let $C_1,\dots,C_n$ subsets of $[1,\ell]$ and $k \in \N$.
Let us consider the alphabet $[0,\ell]$, we define the word $u = 0 1 2 \dots \ell$.
For each $j \in [1,n]$ we let $[1,\ell] \setminus C_j = \set{a_{j,1} < \dots < a_{j,m_j}}$, 
and define $v_j = 0 a_{j,1} \dots a_{j,m_j}$.

We claim that there exists a hitting set of size at most $k$ if and only if 
there exists a formula in $\LTL(\F,\wedge)$ of size at most $3k - 1$ separating $u$ from $v_1,\dots,v_n$.

\vskip1em
Let $H = \set{c_1,\dots,c_k}$ a hitting set of size $k$ with $c_1 < c_2 < \dots < c_k$,
we construct the (non-grounded) fattern induced by $w = c_1 \dots c_k$,
it separates $u$ from $v_1,\dots,v_n$ and has size $3k - 1$.

Conversely, let $\phi$ a formula in $\LTL(\F,\wedge)$ of size $3k - 1$ separating $u$ from $v_1,\dots,v_n$.
Thanks to Lemma~\ref{lem:normalisation_F_and_single_positive} we can assume that $\phi$ is a fattern,
let $w = c_1 \dots c_k$ the non-repeating word it induces.
Necessarily $c_1 < c_2 < \dots < c_k$.
If $\phi$ is grounded then $c_1 = 0$, but then the (non-grounded) fattern induced by $c_2 \dots c_k$
is also separating, so we can assume that $\phi$ is not grounded.
We let $H = \set{c_1,\dots,c_k}$, and argue that $H$ is a hitting set.
Indeed, $H$ is a hitting set if and only if 
for every $j \in [1,n]$ we have $H \cap C_j \neq \emptyset$,
which is equivalent to 
for every $j \in [1,n]$ we have $v_j \not\models \phi$;
indeed for $c_i \in H \cap C_j$ by definition $c_i$ does not appear in $v_j$ so $v_j \not\models \F c_i$.
\end{proof}

%% file: f_x_and_or.tex
\begin{theorem}\label{thm:ltl_F_X_and_or_NP}
The learning problem for $\LTL(\F,\X,\wedge,\vee)$ is in $\NP$.
\end{theorem}

\begin{proof}
Let $u_1,\dots,u_n,v_1,\dots,v_n$ a set of $2n$ words all of length $\ell$.
We note that there always exist a separating formula:
\[
\bigvee_{j \in [1,n]} \bigwedge_{i \in [1,\ell]} \X^{i-1} u_j(i).
\]
This formula\footnote{The formula can be factorised to yield a formula of size $O(n \cdot \ell)$.} has size $O(n \cdot \ell^2)$,
which is polynomial in the size of the input.
A non-deterministic algorithm guesses such a formula of size at most $O(n \cdot \ell^2)$ 
and checks whether it is indeed separating in (deterministic) time $O(n^2 \cdot \ell^3)$.
\end{proof}

We note that the argument applies to any fragment containing $\X,\wedge$, and $\vee$;
in particular this shows that the learning problem for $\LTL = \LTL(\G,\F,\X,\wedge,\vee)$ is in $\NP$.

\subsection*{Hardness result}

We show that the reduction constructed in Section~\ref{sec:X_and} extends to $\LTL(\F, \X,\wedge,\vee)$.

\begin{theorem}\label{thm:hardness_F_X_and_or}
The $\LTL(\F, \X, \wedge, \vee)$ learning problem is $\NP$-hard, and there are no 
$(1 - o(1)) \cdot \log(n)$ polynomial time approximation algorithms unless $\P = \NP$,
even for a single positive word.
\end{theorem}

We prove that the reduction constructed in Theorem~\ref{thm:hardness_X_and} is also a reduction
from set cover to the $\LTL(\F, \X, \wedge, \vee)$ learning problem.

To prove this result we need a reduction lemma for disjunctions, that we state and prove now.
Let $\phi \in \LTL(\F,\X,\wedge,\vee)$, we define $D(\phi) \subseteq \LTL(\F,\X,\wedge)$ by induction:
\begin{itemize}
    \item If $\phi = c$ then $D(\phi) = \set{c}$.
    \item If $\phi = \phi_1 \wedge \phi_2$ then $D(\phi) = \set{\psi_1 \wedge \psi_2 : \psi_1 \in D(\phi_1), \psi_2 \in D(\phi_2)}$.
    \item If $\phi = \phi_1 \vee \phi_2$ then $D(\phi) = D(\phi_1) \cup D(\phi_2)$.
    \item If $\phi = \X \phi'$ then $D(\phi) = \set{\X \psi : \psi \in D(\phi')}$.
    \item If $\phi = \F \phi'$ then $D(\phi)  = \set{\F \psi : \psi \in D(\phi')}$.
\end{itemize}

\begin{lemma}\label{lem:remove_or}
For any $u,v_1,\dots,v_n$, 
if $\phi$ separates $u$ from $v_1,\dots,v_n$,
then there exists $\psi \in D(\phi)$ which separates $u$ from $v_1,\dots,v_n$.
\end{lemma}
\begin{proof}
We proceed by induction on $\phi$.
\begin{itemize}
    \item If $\phi = c$ this is clear.
    \item If $\phi = \phi_1 \wedge \phi_2$ then $D(\phi) = \set{\psi_1 \wedge \psi_2 : \psi_1 \in D(\phi_1), \psi_2 \in D(\phi_2)}$.
    Since $\phi$ separates $u$ from $v_1,\dots,v_n$, there exists $I_1,I_2 \subseteq [1,n]$ 
    such that $I_1 \cup I_2 = [1,n]$, 
    $\phi_1$ separates $u$ from $\set{v_i : i \in I_1}$, and 
    $\phi_2$ separates $u$ from $\set{v_i : i \in I_2}$.
    By induction hypothesis applied to both $\phi_1$ and $\phi_2$ 
    there exists $\psi_1 \in D(\phi_1)$ separating $u$ from $\set{v_i : i \in I_1}$
    and $\psi_2 \in D(\phi_2)$ separating $u$ from $\set{v_i : i \in I_2}$.
    It follows that $\psi_1 \wedge \psi_2$ separates $u$ from $v_1,\dots,v_n$,
    and $\psi_1 \wedge \psi_2 \in D(\phi)$.  
    \item If $\phi = \phi_1 \vee \phi_2$ then $D(\phi) = D(\phi_1) \cup D(\phi_2)$.
    Since $\phi$ separates $u$ from $v_1,\dots,v_n$, 
    either $\phi_1$ or $\phi_2$ does as well; 
    without loss of generality let us say that $\phi_1$ separates $u$ from $v_1,\dots,v_n$.
    The induction hypothesis implies that $\psi_1 \in D(\phi_1)$ separates $u$ from $v_1,\dots,v_n$,
    and $\psi_1 \in D(\phi)$.
    \item The cases $\phi = \X \phi'$ and $\phi = \F \phi'$ follow directly by induction hypothesis.
\end{itemize}
\end{proof}

We now prove Theorem~\ref{thm:hardness_F_X_and_or}.

\begin{proof}
Let $u,v_1,\dots,v_{n+1}$ the words constructed by the reduction.
We claim that if there exists a formula in $\LTL(\F, \X, \wedge, \vee)$ separating $u$ from $v_1,\dots,v_{n+1}$,
then there exists a formula in $\LTL(\X,\wedge)$ separating $u$ from $v_1,\dots,v_{n+1}$ of size smaller than or equal to the original formula.
The proof goes in two steps: 
\begin{itemize}
	\item from $\LTL(\F,\X,\wedge,\vee)$ to $\LTL(\F,\X,\wedge)$;
	\item from $\LTL(\F,\X,\wedge)$ to $\LTL(\X,\wedge)$.
\end{itemize}

Let $\phi \in \LTL(\F,\X,\wedge,\vee)$ separating $u$ from $v_1,\dots,v_{n+1}$. Thanks to Lemma~\ref{lem:remove_or}
there exists $\psi \in D(\phi)$ separating $u$ from $v_1,\dots,v_{n+1}$.
Note that all formulas in $D(\phi)$ are smaller than or equal to $\phi$,
which finishes the proof of the first claim.

Let $\phi \in \LTL(\F,\X,\wedge)$, we define $[\phi] \in \LTL(\X,\wedge)$ by induction:
\begin{itemize}
    \item If $\phi = a$ then $[\phi] = a$.
    \item If $\phi = \phi_1 \wedge \phi_2$ then $[\phi] = [\phi_1] \wedge [\phi_2]$.
    \item If $\phi = \X \phi'$ then $[\phi] = \X [\phi']$.
    \item If $\phi = \F \phi'$ then $[\phi] = [\phi']$.
\end{itemize}

We claim that if $\phi$ separates $u$ from $v_1,\dots,v_{n+1}$, then $[\phi]$ separates $u$ from $v_1,\dots,v_{n+1}$.
To prove this we will establish 3 properties.

\begin{enumerate}
    \item For every word $w$, $w \models [\phi]$ implies $w \models  \phi$.
    \item Let $i \in [2,\ell + 1]$ and $i' \in [1,i-1]$.
    If $u,i \models \phi$ then $v_{n+1},i' \models \phi$.
    \item If $u \models \phi$ and $v_{n+1} \not \models \phi$,
    then $u \models [\phi]$.
\end{enumerate}

Here are the proofs of these three properties.

\begin{enumerate}
    \item By induction on $\phi$, we prove that $w, i \models [\phi]$ implies $w, i \models \phi$.
    \begin{itemize}
        \item If $\phi = a$ then $[\phi] = a$ so the property is trivial.
        \item If $\phi = \phi_1 \wedge \phi_2$ then     $[\phi] = [\phi_1] \wedge [\phi_2]$ so the property follows by induction hypothesis.
        \item If $\phi = \X \phi'$ then $[\phi] = \X [\phi']$ so the property follows by induction hypothesis.
        \item If $\phi = \F \phi'$ then $[\phi] = [\phi']$.
        Assume $w, i \models [\phi]$, meaning $w, i \models [\phi']$. By induction hypothesis this implies that
        $w, i \models \phi'$.
        Now this implies that $w, i \models \F \phi'$ (choose $i' = i$ in the definition of the semantics of $\F$).
    \end{itemize}

    \item Recall that $u = a^{\ell + 1}$ and $v_{n+1} = a^\ell b$.
    By induction on $\phi$, we prove that for all $i \in [2,\ell + 1]$ and $i' \in [1,i-1]$, 
    $u,i \models \phi$ implies $v_{n+1},i' \models \phi$.
    \begin{itemize}
        \item If $\phi \in \set{a,b}$, since $u,i \models \phi$ necessarily $\phi = a$, so $v_{n+1},i' \models \phi$ 
        (indeed $i' \le \ell$ so $v_{n+1}(i') = a$).
        \item If $\phi = \phi_1 \wedge \phi_2$ the property follows by induction hypothesis.
        \item If $\phi = \X \phi'$, we have
        $u,i \models \phi$ if $i+1 \le \ell+1$ and $u,i+1 \models \phi'$.
        By induction hypothesis $v_{n+1},i'+1 \models \phi'$,
        implying that $v_{n+1},i' \models \X \phi' = \phi$.
        \item If $\phi = \F \phi'$, we have
        $u,i \models \phi$ if there exists $i' \in [i,\ell+1]$ such that $u,i' \models \phi'$.
        By induction hypothesis $v_{n+1},i'-1 \models \phi'$,
        with $i'-1 \in [i-1,\ell]$, implying that for $i'' \in [1,i-1]$ we have $v_{n+1},i'' \models \F \phi'$, so $v_{n+1},i'' \models \phi$.
    \end{itemize}
    
    \item By induction on $\phi$, we prove that for all $i \in [1,\ell + 1]$, 
    if $u,i \models \phi$ and $v_{n+1},i \not \models \phi$, then $u,i \models [\phi]$.
    \begin{itemize}
        \item If $\phi \in \{a,b\}$, then $[\phi] = \phi$ so the property is trivial.
        \item If $\phi = \phi_1 \wedge \phi_2$ the property follows by induction hypothesis.
        Indeed, since $u,i \models \phi$ then $u,i \models \phi_1$ and $u,i \models \phi_2$.
        Since $v_{n+1},i \not \models \phi$ then either $v_{n+1},i \not \models \phi_1$ or $v_{n+1},i \not \models \phi_2$.
        Let us consider the first case, the other being symmetric: $v_{n+1},i \not \models \phi_1$.
        By induction hypothesis to $\phi_1$ we get that 
        $v_{n+1}, i \not \models [\phi_1]$.
        Since $[\phi] = [\phi_1] \wedge [\phi_2]$
        this implies that $v_{n+1}, i \not \models [\phi]$.
        \item If $\phi = \X \phi'$ the property follows by induction hypothesis.
        \item If $\phi = \F \phi'$, then $[\phi] = [\phi']$.
        Since $u,i \models \phi$, there exists $i' \in [i,\ell + 1]$ such that $u,i' \models \phi'$.
        The second property implies that necessarily $i' = i$: indeed if $i' > i$ we would have $v_{n+1},i \models \phi'$,
        implying that $v_{n+1},i \models \phi$.
        It follows that $u,i \models \phi'$.
        Since $v_{n+1},i \not \models \phi$ in particular
        $v_{n+1},i \not \models \phi'$.
        By induction hypothesis this implies that $u,i \models [\phi']$, equivalently $u,i \models [\phi]$.
    \end{itemize}
\end{enumerate}

Thanks to these three properties we can show that if $\phi$ separates $u$ from $v_1,\dots,v_{n+1}$, then $[\phi]$ separates $u$ from $v_1,\dots,v_{n+1}$.
Since for each $j \in [1,n+1]$, we have $v_j \not\models \phi$, the first property implies that $v_j \not\models [\phi]$.
Since $u \models \phi$ and $v_{n+1} \not \models \phi$, the third property implies that $u \models [\phi]$.
\end{proof}

%% file: conclusions.tex
Towards stating the remaining most interesting open problems, let us first give an easy dualisation argument.
We define the duals as follows:
\[
\overline{a} = \neg a \qquad 
\overline{\X} = \X \qquad 
\overline{\F} = \G \qquad 
\overline{\G} = \F \qquad 
\overline{\wedge} = \vee \qquad 
\overline{\vee} = \wedge.
\]
For a formula $\phi$ we write $\overline{\phi}$ the formula obtained from $\phi$ by applying $\overline{\cdot}$ inductively.
Clearly, $u \models \phi$ if and only if $u \not\models \overline{\phi}$.
Consequently, $\phi$ separates $u_1,\dots,u_n$ from $v_1,\dots,v_n$ if and only if $\overline{\phi}$ separates $v_1,\dots,v_n$
from $u_1,\dots,u_n$.
Using this duality, $\LTL(\X,\wedge)$ becomes $\LTL(\X,\vee)$, $\LTL(\F,\wedge)$ becomes $\LTL(\G,\vee)$,
and $\LTL(\F,\X,\wedge,\vee)$ becomes $\LTL(\G,\X,\wedge,\vee)$.
Accordingly, all results we obtained for the three fragments apply to their duals.

\vskip1em
We have shown in Section~\ref{sec:F_and} that there is no polynomial time $(1 - o(1)) \cdot \log(n)$-approximation algorithm,
and constructed an (exponential in the number of words) $n$-approximation algorithm.
\begin{open}
Does there exist a polynomial time $O(\log(n))$-approximation algorithm for learning $\LTL(\F,\wedge)$?
\end{open}

We have proved that the learning problem is $\NP$-complete for the fragments 
$\LTL(\X,\wedge), \LTL(\F,\wedge)$, $\LTL(\F,\X,\wedge,\vee)$, and their duals.
The reduction used for proving the last result does not extend to full $\LTL$ (indeed $\G a$ separates $u$ from $v_1,\dots,v_{n+1}$).

\begin{open}
Is the learning problem $\NP$-complete for full $\LTL$?
\end{open}